\RequirePackage{amsmath}
\documentclass[runningheads,a4paper]{llncs}

\usepackage{amssymb,amsmath,tikz,changepage}
\usepackage{graphicx}

\usepackage{url}
\usepackage{enumerate}

\usepackage{algorithm}
\usepackage{algpseudocode}

\usepackage{subcaption}
\makeatletter
\newcommand{\text@hyphens}{\mathcode`\-=`\-\relax}
\newcommand{\id}[1]{\ensuremath{\mathit{\text@hyphens#1}}}
\makeatother

\begin{document}

\mainmatter  
\title{Approximate Maximin Shares for Groups of Agents}
\titlerunning{Approximate Maximin Shares for Groups of Agents}
\author{Warut Suksompong%
}
\authorrunning{W. Suksompong}
\institute{Department of Computer Science, Stanford University\\
353 Serra Mall, Stanford, CA 94305, USA\\
\email{warut@cs.stanford.edu}\\
}
\maketitle

\begin{abstract}
We investigate the problem of fairly allocating indivisible goods among interested agents using the concept of maximin share. Procaccia and Wang showed that while an allocation that gives every agent at least her maximin share does not necessarily exist, one that gives every agent at least $2/3$ of her share always does. In this paper, we consider the more general setting where we allocate the goods to \emph{groups} of agents. The agents in each group share the same set of goods even though they may have conflicting preferences. For two groups, we characterize the cardinality of the groups for which a positive approximation of the maximin share is possible regardless of the number of goods. We also show settings where an approximation is possible or impossible when there are several groups.
\end{abstract}

\section{Introduction}
\label{sec:intro}

We consider the problem of fairly allocating indivisible goods to interested agents, a task that occurs frequently in the society and has been a subject of study for decades in economics and more recently in computer science. Several notions of fairness have been proposed to this end; the most popular ones include \emph{envy-freeness} \cite{Foley67,Varian74} and \emph{proportionality} \cite{Steinhaus48}. Envy-freeness stipulates that each agent likes her bundle at least as much as that of any other agent, while proportionality requires that each agent receive her proportional share in the allocation. While an allocation satisfying both notions can be obtained when we deal with divisible goods such as cake or land, this is not the case for indivisible goods like houses or cars. Indeed, if there is one indivisible good and several agents, some agent is necessarily left empty-handed and neither of the notions can be satisfied. In fact, the same example shows that even a multiplicative approximation of these notions cannot be guaranteed.\footnote{As a result of the lack of existence guarantee for these fairness criteria, statements on the asymptotic existence when utilities are drawn from distributions have been shown \cite{AmanatidisMaNi15,DickersonGoKa14,KurokawaPrWa16,ManurangsiSu17,Suksompong16-2}.} 

A notion that was designed to fix this problem and has been a subject of much interest in the last few years is called the \emph{maximin share}, first introduced in this context by Budish \cite{Budish11} based on earlier concepts by Moulin \cite{Moulin90}. The maximin share of an agent can be determined as follows: If there are $n$ agents in total, let the agent of interest partition the goods into $n$ bundles, knowing that she will get the bundle that she likes least. The maximum utility that she can obtain via this procedure is her maximin share. The intuition behind the maximin share is that the agent could feel entitled to this share, since she is already taking the least preferable bundle (according to her own preferences) from the partition. In other words, the indivisibility of the goods cannot be an excuse not to give her this share. An allocation is said to satisfy the \emph{maximin share criterion} if every agent receives a utility no less than her maximin share. 

While an allocation satisfying the maximin share criterion exists in several restricted settings \cite{BouveretLe16} as well as asymptotically \cite{AmanatidisMaNi15,KurokawaPrWa16}, Procaccia and Wang \cite{ProcacciaWa14} showed through rather intricate examples that, somewhat surprisingly, such an allocation does not always exist in general. As a result, the same authors turned the focus to approximation and showed that an allocation that yields a $2/3$-approximation of the maximin share to every agent always exists, and Amanatidis et al. \cite{AmanatidisMaNi15} exhibited an efficient algorithm that computes such an allocation. Amanatidis et al. also presented a simpler algorithm with an approximation ratio of 1/2. Barman and Krishna Murthy \cite{BarmanKr17} gave a simpler efficient $2/3$-approximation algorithm and moreover initiated the study of maximin share for agents with submodular (as opposed to additive) valuation functions. Gourv\`{e}s and Monnot \cite{GourvesMo17} extended the problem to the case where the goods satisfy a matroidal constraint. Farhadi et al. \cite{FarhadiGhHa17} studied the maximin share criterion when agents have unequal entitlements. Recently, Ghodsi et al. \cite{GhodsiHaSe17} improved the approximation ratio to $3/4$. The question of truthfulness for maximin share approximation algorithms has also been explored \cite{AmanatidisBiCh17,AmanatidisBiMa16}. Besides its theoretical appeal, the maximin share has been used in applications including the popular fair division website Spliddit \cite{GoldmanPr14}.

In this paper, we apply the concept of maximin share to a more general setting of fair division in which goods are allocated not to individual agents, but rather to groups of agents who can have varying preferences on the goods. Several practical situations involving fair division fit into this model. For instance, an outcome of a negotiation between countries may have to be approved by members of the cabinets of each country. It could well be the case that one member of a cabinet of a country thinks that the division is fair while another does not. Similarly, in a divorce case, it is not hard to imagine that different members of the family on the husband side and the wife side have varying opinions on a proposed settlement. Another example is a large company or university that needs to divide its resources among competing groups of agents (e.g., departments in a university). The agents in each group have different and possibly misaligned interests; the professors who perform theoretical research may prefer more whiteboards and open space in the department building, while those who engage in experimental work are more likely to prefer laboratories. These situations cannot be modeled by the traditional fair division setting where each recipient of a bundle of goods is represented by a single preference. The added element of having several agents in the same group receiving the same bundle of goods corresponds to the social choice framework of aggregating individual preferences to reach a collective decision \cite{ArrowSeSu02}. Manurangsi and Suksompong \cite{ManurangsiSu17} recently investigated the asymptotic existence of fair divisions under this setting, while Segal-Halevi and Nitzan \cite{SegalhaleviNi15,SegalhaleviNi16} considered fairness in the allocation of \emph{divisible} goods to groups. A related setting, where a subset of indivisible goods is allocated to a (single) group of agents, has also been studied \cite{ManurangsiSu17-2,Suksompong16}.

\subsection{Our Results} 

We extend the maximin share notion to groups in a natural way by calculating the maximin share for each agent using the number of groups instead of the number of agents. When there are two groups, we completely determine the cardinality of agents in the groups for which it is possible to approximate the maximin share within a positive factor that depends only on the number of agents and not on the number of goods. In particular, an approximation is possible when one of the groups contain a single agent, when both groups contain two agents, or when the groups contain three and two agents respectively. In all other cases, no approximation is possible in a strong sense: There exists an instance with only four goods in which some agent with positive maximin share necessarily gets zero utility. The results for the setting with two groups are presented in Section \ref{sec:twogroups} and summarized in Table \ref{table:twogroups}. Even though we leave a gap between the lower and upper bounds of the approximation ratio, the reader should bear in mind that this gap remains even for the previously studied setting where goods are allocated to individual agents. Indeed, the gap between $3/4$ and $1-o(1)$ for maximin share approximation has not been closed despite several works in this direction, while in the simplest case with three agents, the gap remains between 8/9 and $1-o(1)$ \cite{AmanatidisMaNi15,GhodsiHaSe17,GourvesMo17,KurokawaPrWa16,ProcacciaWa14}.\footnote{When there are two agents, it is not hard to see that a cut-and-choose protocol guarantees the full maximin share for both agents.} Thus, despite its relatively simple definition, determining the best approximation ratio for the maximin share is perhaps a harder problem than it might seem at first glance. In addition, although the case of two groups might seem like a rather restricted case, recall that fair division between two \emph{agents}, which is even more restricted, has enjoyed significant attention in the literature (e.g., \cite{BramsFi00,BramsKiKl12,BramsKiKl14}). 

In Section \ref{sec:manygroups}, we generalize to the setting with several groups of agents. On the positive side, we show that a positive approximation is possible if only one group contains more than a single agent. On the other hand, we show on the negative side that when all groups contain at least two agents and one group contains at least five agents, it is possible that some agent with positive maximin share will be forced to obtain zero utility, which means that there is no hope of obtaining an approximation in this case.

\section{Preliminaries}
\label{sec:prelim}

Let $G=\{g_1,g_2,\dots,g_m\}$ denote the set of goods, and $A$ the set of agents. The agents are partitioned into $k$ groups $A_1,\dots,A_k$. Group $A_i$ contains $n_i$ agents; denote by $a_{ij}$ the $j$th agent in group $A_i$. The agents in each group will be collectively allocated a subset of goods $G$; suppose that the agents in group $A_k$ receive the subset $G_k$. 

Each agent $a_{ij}$ has some nonnegative utility $u_{ij}(g)$ for every good $g\in G$. We assume that the agents are endowed with \emph{additive} utility functions, i.e., $u_{ij}(G')=\sum_{g\in G'}u_{ij}(g)$ for any agent $a_{ij}\in A$ and any subset of goods $G'\subseteq G$. This assumption is commonly made in the fair division literature, especially the literature that deals with the maximin share; it provides a reasonable tradeoff between simplicity and expressiveness \cite{AmanatidisMaNi15,BouveretLe16,KurokawaPrWa16,ProcacciaWa14}. Let $\textbf{u}_{ij}=(u_{ij}(g_1),u_{ij}(g_2),\dots,u_{ij}(g_m))$ be the utility vector of agent $a_{ij}$. We refer to a setting with groups of agents, goods, and utility functions as an \emph{instance}.

We now define the maximin share. Let $K=\{1,2,\dots,k\}$. Suppose that in a partition $\mathcal{G}'$ of the goods, the agents in group $A_k$ receive the subset $G'_k$. 

\begin{definition}
\label{def:maximin}
The \emph{maximin share} of agent $a_{ij}$ is defined as 
\[\max_{\mathcal{G}'}\min_{k\in K}u_{ij}(G'_k),\]
 where the maximum ranges over all (complete) partitions $\mathcal{G}'$ of the goods in $G$. Any partition for which this maximum is attained is called a \emph{maximin partition} of agent $a_{ij}$.

An allocation $\mathcal{G}$ of goods to the groups is said to satisfy the \emph{maximin share criterion} if every agent obtains at least her maximin share in $\mathcal{G}$.
\end{definition}

As an example, suppose that there are two groups and four goods, and an agent's utilities for the goods are given by $u(g_1)=6, u(g_2)=3$, and $u(g_3)=u(g_4)=2$. The maximin share of the agent is 6, as can be seen from the partition $(\{g_1\},\{g_2,g_3,g_4\})$. This partition is the unique maximin partition for the agent.

It follows directly from Definition \ref{def:maximin} that the maximin share of an agent $a_{ij}$ is at most $u_{ij}(G)/k$. In addition, any envy-free or proportional allocation also satisfies the maximin share criterion.\footnote{Assuming we extend the notions of envy-freeness and proportionality to groups in a similar manner as we do for the maximin share.} Since an allocation satisfying the maximin share criterion does not always exist even for three groups with one agent each \cite{ProcacciaWa14}, we will be interested in one that \emph{approximates} the maximin share criterion, i.e., gives every agent at least a multiplicative factor $\alpha$ of her maximin share, for some $\alpha\in(0,1)$.

\section{Two Groups of Agents}
\label{sec:twogroups}

In this section, we consider the setting where there are two groups of agents and characterize the cardinality of the groups for which a positive approximation of the maximin share is possible regardless of the number of goods. In particular, suppose that the two groups contain $n_1$ and $n_2$ agents, where we assume without loss of generality that $n_1\geq n_2$. Then a positive approximation is possible when $n_2=1$ as well as when $(n_1,n_2)=(2,2)$ or $(3,2)$.  The results are summarized in Table \ref{table:twogroups}.

\begin{table}
\centering
    \begin{tabular}{| c | c |}
    \hline
     \textbf{Number of agents} & \textbf{Approximation ratio} \\ \hline
     $(n_1,n_2)=(1,1)$ & $\alpha=1$ (Cut-and-choose protocol; see, e.g., \cite{BouveretLe16}) \\ \hline
     $(n_1,n_2)=(2,1)$ & $2/3\leq\alpha\leq 3/4$ (Theorem \ref{thm:tworoundrobin}) \\ \hline
     $n_2=1$ & $2/(n_1+1)\leq \alpha\leq 1/\left\lfloor \left\lfloor\sqrt{2n_1}\right\rfloor/2\right\rfloor$ (Corollary \ref{cor:manyonefloor})   \\ \hline
     $(n_1,n_2)=(2,2)$ & $1/8\leq\alpha\leq 1/2$ (Theorem \ref{thm:twotwo}) \\ \hline
     $(n_1,n_2)=(3,2)$ & $1/16\leq\alpha\leq 1/2$ (Theorem \ref{thm:threetwo}) \\ \hline
     $n_1\geq 4, n_2\geq 2$ & $\alpha=0$ (Proposition \ref{prop:fourtwo}) \\ \hline
     $n_1,n_2\geq 3$ & $\alpha=0$ (Proposition \ref{prop:threethree}) \\ 
    \hline
    \end{tabular}
    \vspace{5mm}
    \caption{Values of the best possible approximation ratio, denoted by $\alpha$, for the maximin share when there are two groups with $n_1\geq n_2$ agents. The approximation ratios hold regardless of the number of goods.}
    \label{table:twogroups}
\end{table}

\subsection{Large Number of Agents: No Possible Approximation}

We start by showing that when the numbers of agents in the groups are large enough, no approximation of the maximin share is possible. Observe that if we prove that a maximin share approximation is not possible for groups with $n_1$ and $n_2$ agents, then it is also not possible for groups with $n_1'\geq n_1$ and $n_2'\geq n_2$ agents, since we would still need to fulfill the approximation for the first $n_1$ and $n_2$ agents in the respective groups.

\begin{proposition}
\label{prop:fourtwo}
If $n_1\geq 4$ and $n_2\geq 2$, then there exists an instance in which some agent with nonzero maximin share necessarily receives zero utility.
\end{proposition}

\begin{proof}
Assume that $n_1=4$ and $n_2=2$, and suppose that there are four goods. The utilities of the agents in the first group are $\textbf{u}_{11}=(0,1,0,1)$, $\textbf{u}_{12}=(0,1,1,0)$, $\textbf{u}_{13}=(1,0,0,1)$, and $\textbf{u}_{14}=(1,0,1,0)$, while the utilities of those in the second group are $\textbf{u}_{21}=(1,1,0,0)$ and $\textbf{u}_{22}=(0,0,1,1)$. 

In this example, every agent has a maximin share of 1. To guarantee nonzero utility for the agents in the second group, we must allocate at least one of the first two goods and at least one of the last two goods to the group. But this implies that some agent in the first group receives zero utility. \qed
\end{proof}

\begin{proposition}
\label{prop:threethree}
If $n_1, n_2\geq 3$, then there exists an instance in which some agent with nonzero maximin share necessarily receives zero utility.
\end{proposition}

\begin{proof}
Assume that $n_1=n_2=3$, and suppose that there are four goods. The utilities of the agents in the first group are $\textbf{u}_{11}=(0,1,0,1)$, $\textbf{u}_{12}=(1,0,0,1)$, and $\textbf{u}_{13}=(1,0,1,0)$, while the utilities of those in the second group are $\textbf{u}_{21}=(1,1,0,0)$, $\textbf{u}_{22}=(0,0,1,1)$, and $\textbf{u}_{23}=(1,0,0,1)$. 

In this example, every agent has a maximin share of 1. To guarantee nonzero utility for the agents in the second group, we must allocate at least one of the first two goods and at least one of the last two goods to the group. If we allocate at least three goods to the second group, some agent in the first group is left with zero utility. Else, if we allocate exactly two goods to the second group, we may allocate goods $\{g_1,g_3\}$, $\{g_1,g_4\}$, or $\{g_2,g_4\}$, which leaves goods $\{g_2,g_4\}$, $\{g_2,g_3\}$, or $\{g_1,g_3\}$ to the first group, respectively. But in each of these cases, some agent in the first group receives zero utility. \qed
\end{proof}

\subsection{Approximation via Modified Round-Robin Algorithm}

When both groups contain a single agent, it is known that a simple ``cut-and-choose'' protocol similar to a famous cake-cutting protocol yields the full maximin share for both agents (see, e.g., \cite{BouveretLe16}). It turns out that as soon as at least one group contains more than one agent, the full maximin share can no longer be guaranteed. We next consider the simplest such case where the groups contain one and two agents, respectively. The maximin share approximation algorithm for this case is similar to the modified round-robin algorithm that yields a $1/2$-approximation for an arbitrary number of agents \cite{AmanatidisMaNi15}, but we will need to make some adjustments to handle more than one agent being in the same group. For the algorithm, we will need the following two lemmas, which admit rather straightforward proofs.

\begin{lemma}[\cite{AmanatidisMaNi15}]
\label{lem:roundrobin}
Suppose that each group contains one agent. Consider a round-robin algorithm in which the agents take turns taking their favorite good from the remaining goods. In the resulting allocation, the envy that an agent has toward any other agent is at most the maximum utility of the former agent for any single good. Moreover, if an agent is ahead of another agent in the round-robin ordering, then the former agent has no envy toward the latter agent.
\end{lemma}

\begin{lemma}[\cite{AmanatidisMaNi15,BouveretLe16}]
\label{lem:addonemfs}
Given an arbitrary instance in which each group contains one agent. If we allocate an arbitrary good to an agent as her only good, then the maximin share of any remaining agent with respect to the remaining goods does not decrease. 
\end{lemma}

We are now ready to handle the case with one and two agents in the groups.

\begin{theorem}
\label{thm:tworoundrobin}
Let $(n_1,n_2)=(2,1)$, and suppose that $\alpha$ is the best possible approximation ratio for the maximin share. Then $2/3\leq \alpha\leq 3/4$.
\end{theorem}

\begin{proof}
We first show the upper bound. Suppose that there are four goods. The utilities of the agents in the first group for the goods are $\textbf{u}_{11}=(3,1,2,2)$ and $\textbf{u}_{12}=(2,3,2,1)$, while the utilities of the agent in the second group are $\textbf{u}_{21}=(3,2,2,1)$. 

In this example, every agent has a maximin share of 4. We will show that any allocation gives some agent a utility of at most 3. Note that an allocation that would give every agent a utility of at least 4 must allocate two goods to both groups. If the fourth and one of the second and third goods are allocated to the second group, the agent gets a utility of 3. Otherwise, one can check that one of the agents in the first group gets a utility of 3.

Next, we exhibit an algorithm that guarantees each agent a $2/3$ fraction of her maximin share. Since we do not engage in interpersonal comparisons of utilities, we may assume without loss of generality that every agent has utility 1 for the whole bundle of goods. Since the maximin share of an agent is always at most $1/2$, it suffices to allocate a bundle worth at least $1/3$ to her.

If some good is worth at least $1/3$ to $a_{21}$, let her take that good. By Lemma \ref{lem:addonemfs}, the maximin shares of $a_{11}$ and $a_{12}$ do not decrease. Since they receive all of the remaining goods, they obtain their maximin share. Hence we may now assume that no good is worth at least $1/3$ to $a_{21}$. We let each of $a_{11}$ and $a_{12}$, in arbitrary order, take a good worth at least $1/3$ if there is any. There are three cases.
\begin{itemize}
\item Both of them take a good. Then each of them gets a utility of at least $1/3$. Since every good is worth less than $1/3$ to $a_{21}$, she also gets a utility of at least $1-1/3-1/3=1/3$ from the remaining goods.
\item One of them takes a good and the other does not. Assume without loss of generality that $a_{11}$ is the agent who takes a good. We run the round-robin algorithm on $a_{12}$ and $a_{21}$ using the remaining goods, starting with $a_{21}$. Since every good is worth less than $1/3$ to $a_{21}$, the value of the whole bundle of goods minus the good that $a_{11}$ takes is at least $2/3$. By Lemma \ref{lem:roundrobin}, $a_{21}$ gets a utility of at least $1/2\times 2/3 = 1/3$. Similarly, the envy of $a_{12}$ toward $a_{21}$ is at most the maximum utility of $a_{12}$ for a good allocated during the round-robin algorithm, which is at most $1/3$. This implies that $a_{21}$'s bundle is worth at most $2/3$ to $a_{12}$, and hence $a_{12}$'s bundle in the final allocation (i.e., her bundle from the round-robin algorithm combined with the good that $a_{11}$ takes) is worth at least $1/3$ to her. 
\item Neither of them takes a good. We run the round-robin algorithm on all three agents, starting with $a_{21}$. By Lemma \ref{lem:roundrobin}, $a_{21}$ gets a utility of at least $1/3$. The envy of $a_{11}$ toward $a_{21}$ is at most the maximum utility of $a_{11}$ for a good, which is at most $1/3$. Hence $a_{21}$'s bundle is worth at most $2/3$ to $a_{11}$, which means that $a_{11}$'s bundle in the final allocation (i.e., her bundle combined with $a_{12}$'s bundle) is worth at least $1/3$ to her. An analogous argument holds for $a_{12}$.
\end{itemize}
This covers all three possible cases. \qed
\end{proof}

Next, we generalize to the setting where the first group contains an arbitrary number of agents while the second group contains a single agent. In this case, an algorithm similar to that in Theorem \ref{thm:tworoundrobin} can be used to obtain a constant factor approximation when the number of agents is constant. In addition, we show that the approximation ratio necessarily degrades as the number of agents grows.

\begin{algorithm}
\caption{Algorithm for approximate maximin share when the groups contain $n_1\geq 2$ and $n_2=1$ agents (Theorem \ref{thm:manyone}).}
\label{alg:algomanyone}
\begin{algorithmic}[1]
\Procedure{Approximate\textendash Maximin\textendash Share\textendash 1}{}
\If{Agent $a_{21}$ values some good $g$ at least $\frac{1}{n_1+1}$}
\State Allocate $g$ to $a_{21}$ and the remaining goods to the first group.
\Else
\State Let each agent in the first group, in arbitrary order, take a good worth at least $\frac{1}{n_1+1}$ to her if there is any.
\State Allocate the remaining goods to the agents who have not taken a good using the round-robin algorithm, starting with $a_{21}$.
\EndIf
\EndProcedure
\end{algorithmic}
\end{algorithm}

\begin{theorem}
\label{thm:manyone}
Let $n_1\geq 2$ and $n_2=1$, and suppose that $\alpha$ is the best approximation ratio for the maximin share. Then $\frac{2}{n_1+1}\leq \alpha\leq \frac{1}{\lfloor f(n_1)/2\rfloor}$, where $f(n_1)$ is the largest integer such that $\binom{f(n_1)}{2}\leq n_1$.
\end{theorem}

\begin{proof}
We first show the upper bound. Let $l=f(n_1)$, and suppose that there are $l$ goods. Let $\binom{l}{2}$ of the agents in the first group positively value a distinct set of two goods. In particular, each of them has utility 1 for both goods in their set and 0 for the remaining goods. Let the agent in the second group have utility 1 for all goods.

In this example, each agent in the first group has a maximin share of 1, while the agent $a_{21}$ has a maximin share of $\lfloor l/2\rfloor$. To guarantee nonzero utility for the $l$ agents in the first group, we must allocate all but at most one good to the group, leaving at most one good for the second group. So the agent in the second group obtains at most a $1/\lfloor l/2\rfloor$ fraction of her maximin share. 

An algorithm that guarantees a $\frac{2}{n_1+1}$-approximation of the maximin share (Algorithm \ref{alg:algomanyone}) is similar to that for the case $n_1=2$ (Theorem \ref{thm:tworoundrobin}). Again, we normalize the utility of each agent for the whole set of goods to 1. First, let $a_{21}$ take a good worth at least $\frac{1}{n_1+1}$ to her if there is any. If she takes a good, we allocate the remaining goods to the first group and are done by Lemma \ref{lem:addonemfs}. Else, we let each of the agents in the first group, in arbitrary order, take a good worth at least $\frac{1}{n_1+1}$ to her if there is any. After that, we run the round-robin algorithm on the agents who have not taken a good, starting with $a_{21}$.

Suppose that $r$ agents in the first group take a good. Each of them obtains a utility of at least $\frac{1}{n_1+1}$. The remaining goods, which are allocated by the round-robin algorithm, are worth a total of at least $\frac{n_1+1-r}{n_1+1}$ to $a_{21}$. Since there are $n_1+1-r$ agents who participate in the round-robin algorithm, and $a_{21}$ is the first to choose, she obtains utility at least $\frac{1}{n_1+1-r}\cdot\frac{n_1+1-r}{n_1+1}=\frac{1}{n_1+1}$. Finally, by Lemma \ref{lem:roundrobin}, each agent in the first group who does not take a good in the first stage has envy at most $\frac{1}{n_1+1}\leq \frac13$ toward $a_{21}$. Hence for such an agent, $a_{21}$'s bundle is worth at most $2/3$, and so the bundle allocated to the first group is worth at least $\frac13\geq \frac{1}{n_1+1}$. \qed
\end{proof}

Algorithm \ref{alg:algomanyone} can be implemented in time polynomial in the number of agents and goods. Also, since $\binom{\lfloor\sqrt{2n_1}\rfloor}{2}\leq \frac{(\sqrt{2n_1})^2}{2}=n_1$, we have the following corollary.

\begin{corollary}
\label{cor:manyonefloor}
Let $n_1\geq 2$ and $n_2=1$, and suppose that $\alpha$ is the best approximation ratio for the maximin share. Then $\frac{2}{n_1+1}\leq \alpha\leq \frac{1}{\left\lfloor \left\lfloor\sqrt{2n_1}\right\rfloor/2\right\rfloor}.$ 
\end{corollary}

\subsection{Approximation via Maximin Partitions}

We now consider the two remaining cases, $(n_1,n_2)=(2,2)$ and $(3,2)$. We show that in both cases, a positive approximation is also possible. However, the algorithms for these two cases will rely on a different idea than the previous algorithms. These positive results provide a clear distinction between the settings where it is possible to approximate the maximin share and those where it is not. For the former settings, the maximin share can be approximated within a positive factor independent of the number of goods. On the other hand, for the latter settings, there exist instances in which some agent with positive maximin share necessarily gets zero utility even when there are only four goods (Propositions \ref{prop:fourtwo} and \ref{prop:threethree}), and therefore no approximation is possible even if we allow dependence on the number of goods.

\begin{algorithm}
\caption{Algorithm for approximate maximin share when the groups contain $n_1=n_2=2$ agents (Theorem \ref{thm:twotwo}).}
\label{alg:algotwotwo}
\begin{algorithmic}[1]
\Procedure{Approximate\textendash Maximin\textendash Share\textendash 2}{}
\For{each agent $a_{ij}$}
\State Compute her maximin partition $(G_{ij},G\backslash G_{ij})$.
\EndFor
\State Partition $G$ into 16 subsets $(H_1,H_2,\dots,H_{16})$ according to whether each good belongs to $G_{ij}$ or $G\backslash G_{ij}$ for each $1\leq i,j\leq 2$.
\If{Some subset $H_p$ is \emph{important} (i.e., of value at least $1/8$ of the agent's maximin share) to both $a_{11}$ and $a_{12}$}
\State Allocate $H_p$ to the first group and the remaining goods to the second group.
\Else
\State Suppose that $H_p,H_q$ are important to $a_{11}$ and $H_r,H_s$ to $a_{12}$.
\State Find a pair from $(H_p,H_r),(H_p,H_s),(H_q,H_r),(H_q,H_s)$ that does not coincide with the important subsets for an agent in the second group.
\State Allocate that pair of subsets to the first group and the remaining goods to the second group.
\EndIf
\EndProcedure
\end{algorithmic}
\end{algorithm}

\begin{theorem}
\label{thm:twotwo}
Let $(n_1,n_2)=(2,2)$, and suppose that $\alpha$ is the best possible approximation ratio for the maximin share. Then $1/8\leq \alpha\leq 1/2$.
\end{theorem}

\begin{proof}
We first show the upper bound. Suppose that there are four goods. The utilities of the agents in the first group for the goods are $\textbf{u}_{11}=(0,2,1,1)$ and $\textbf{u}_{12}=(2,0,1,1)$, while the utilities of those in the second group are $\textbf{u}_{21}=(1,1,0,0)$ and $\textbf{u}_{22}=(0,0,1,1)$. 

In this example, both agents in the first group have a maximin share of 2, and both agents in the second group have a maximin share of 1. To ensure nonzero utility for the agents in the second group, we must allocate at least one of the first two goods and at least one of the last two goods to the group. However, this implies that some agent in the first group gets a utility of at most 1.

Next, we exhibit an algorithm that yields a $1/8$-approximation of the maximin share (Algorithm~\ref{alg:algotwotwo}). For each agent $a_{ij}$, let $(G_{ij},G\backslash G_{ij})$ be one of her maximin partitions. By definition, we have that both $u_{ij}(G_{ij})$ and $u_{ij}(G\backslash G_{ij})$ are at least the maximin share of $a_{ij}$. Let $(H_1,H_2,\dots,H_{16})$ be a partition of $G$ into 16 subsets according to whether each good belongs to $G_{ij}$ or $G\backslash G_{ij}$ for each $1\leq i,j\leq 2$; in other words, for every $i,j$, the goods in each set $H_k$ either all belong to $G_{ij}$ or all belong to $G\backslash G_{ij}$. By the pigeonhole principle, for each agent $a_{ij}$, among the eight sets $H_k$ whose union is $G_{ij}$, the set that she values most gives her a utility of at least $1/8$ of her maximin share; call this set $H_p$. Likewise, we can find a set $H_q\subseteq G\backslash G_{ij}$ that the agent values at least $1/8$ of her maximin share. We call these subsets \emph{important} to $a_{ij}$. It suffices for every agent to obtain a subset that is important to her.

If some subset $H_p$ is important to both $a_{11}$ and $a_{12}$, we allocate that subset to the first group and the remaining goods to the second group. Since each agent in the second group has at least two important subsets, and only one is taken away from them, this yields the desired guarantee. Else, two subsets $H_p,H_q$ are important to $a_{11}$ and two other subsets $H_r,H_s$ are important to $a_{12}$. We will assign one of the pairs $(H_p,H_r),(H_p,H_s),(H_q,H_r),(H_q,H_s)$ to the first group. If a pair does not work, that means that some agent in the second group has exactly that pair as her important subsets. But there are four pairs and only two agents in the second group, hence some pair must work. \qed
\end{proof}

We briefly discuss the running time of Algorithm \ref{alg:algotwotwo}. The algorithm can be implemented efficiently except for one step: computing a maximin partition of each agent. This step is NP-hard even when the two agents have identical utility functions by a straightforward reduction from the partition problem. Nevertheless, Woeginger \cite{Woeginger97} showed that a PTAS for the problem exists.\footnote{Woeginger also showed that an FPTAS for this problem does not exist unless $\text{P}=\text{NP}$.} Using the PTAS, we can compute an approximate maximin partition instead of an exact one and obtain a $\left(1/8-\epsilon\right)$-approximate algorithm for the maximin share in time polynomial in the number of goods for any constant $\epsilon>0$.

A similar idea can be used to show that a positive approximation of the maximin share is possible when $(n_1,n_2)=(3,2)$.

\begin{theorem}
\label{thm:threetwo}
Let $(n_1,n_2)=(3,2)$, and suppose that $\alpha$ is the best possible approximation ratio for the maximin share. Then $1/16\leq \alpha\leq 1/2$.
\end{theorem}

\begin{proof}
The upper bound follows from Theorem \ref{thm:twotwo} and the observation preceding Proposition \ref{prop:fourtwo}. 

For the lower bound, compute the maximin partition for each agent, and partition $G$ into 32 subsets according to which part of the partition of each agent a good belongs to. For each agent, at least 2 of the subsets are \emph{important}, i.e., of value at least $1/16$ of her maximin share. If some subset is important to both $a_{21}$ and $a_{22}$, allocate that subset to them and the remaining goods to the first group. Otherwise, we can allocate some pair of important subsets to the second group using a similar argument as in Theorem \ref{thm:twotwo}. \qed
\end{proof}

\subsection{Experimental Results}

To complement our theoretical results, we ran computer experiments to see the extent to which it is possible to approximate the maximin share in random instances. For $(n_1,n_2)=(2,2)$ and $(3,2)$, we generated 100000 random instances where there are four goods and the utility of each agent for each good is drawn independently and uniformly at random from the interval $[0,1]$. The results are shown in Table \ref{table:experiment1}. An approximation ratio of $0.9$ can be guaranteed in over $90\%$ when there are two agents in each group, and in over $80\%$ when there are three agents in one group and two in the other. In other words, an allocation that ``almost'' satisfies the maximin criterion can be found in a large majority of the instances. However, the proportion drops significantly to around $70\%$ and $50\%$ if we demand that the partition yield the full maximin share to the agents, indicating that this is a much more stringent requirement. We also ran the experiment on instances where the utilities are drawn from an exponential distribution and from a log-normal distribution. As shown in Tables \ref{table:experiment2} and \ref{table:experiment3}, the number of instances for which the (approximate) maximin criterion is satisfied is lower for both distributions than for the uniform distribution. This is to be expected since the utilities are less spread out, meaning that conflicts are more likely to occur. The heavy drop as we increase the requirement from $\alpha\geq 0.9$ to $\alpha\geq 1$ is present for these distributions as well.

We also remark here that the case with four goods seems to be the hardest case for maximin approximation. Indeed, with two goods an allocation yielding the full maximin share always exists, with three goods the maximin share is low since any partition leaves at most one good to one group, and with more than four goods there are more allocations and therefore more possibilities to exploit the differences between the utilities of the agents. This intuition aligns with the fact that the instances we use to show the approximation lower bounds in this paper all involve exactly four goods.

\begin{table}
\centering
    \begin{subtable}{\textwidth}
    \centering
    \begin{tabular}{| c | c | c | c | c | c | c |}
    \hline
      & $\alpha \geq 0.5$ & $\alpha \geq 0.6$ & $\alpha \geq 0.7$ & $\alpha \geq 0.8$ & $\alpha \geq 0.9$ & $\alpha \geq 1$ \\ \hline
    $(n_1,n_2)=(2,2)$ & 100000 & 100000 & 99937 & 98803 & 92015 & 69248 \\ \hline
    $(n_1,n_2)=(3,2)$ & 100000 & 99997 & 99672 & 96174 & 81709 & 49386 \\
    \hline
    \end{tabular}
    \caption{The uniform distribution over the interval $[0,1]$.}
    \vspace{2mm}
    \label{table:experiment1}
    \end{subtable}
    \begin{subtable}{\textwidth}
    \centering
    \begin{tabular}{| c | c | c | c | c | c | c |}
    \hline
      & $\alpha \geq 0.5$ & $\alpha \geq 0.6$ & $\alpha \geq 0.7$ & $\alpha \geq 0.8$ & $\alpha \geq 0.9$ & $\alpha \geq 1$ \\ \hline
    $(n_1,n_2)=(2,2)$ & 100000 & 99982 & 99280 & 94464 & 80683 & 55833 \\ \hline
    $(n_1,n_2)=(3,2)$ & 100000 & 99827 & 97295 & 86293 & 63914 & 36626 \\
    \hline
    \end{tabular}
    \caption{The exponential distribution with mean 1.}
    \vspace{2mm}
    \label{table:experiment2}
    \end{subtable}
    \begin{subtable}{\textwidth}
    \centering
    \begin{tabular}{| c | c | c | c | c | c | c |}
    \hline
      & $\alpha \geq 0.5$ & $\alpha \geq 0.6$ & $\alpha \geq 0.7$ & $\alpha \geq 0.8$ & $\alpha \geq 0.9$ & $\alpha \geq 1$ \\ \hline
    $(n_1,n_2)=(2,2)$ & 100000 & 99990 & 99220 & 92658 & 74966 & 55768 \\ \hline
    $(n_1,n_2)=(3,2)$ & 100000 & 99895 & 97159 & 82918 & 57068 & 36802 \\
    \hline
    \end{tabular}
    \caption{The log-normal distribution with parameters $\mu=0$ and $\sigma=1$ for the associated normal distribution.}
    \label{table:experiment3}
    \end{subtable}
\caption{Experimental results showing the number of instances, out of 100000, for which the respective maximin approximation ratio is achievable by some allocation of goods when utilities are drawn independently from the specified probability distribution.}
\end{table}


\section{Several Groups of Agents}
\label{sec:manygroups}

In this section, we consider a more general setting where there are several groups of agents. We show that when only one group contains more than a single agent, a positive approximation is possible independent of the number of goods. On the other hand, when all groups contain at least two agents and one group contain at least five agents, no approximation is possible in a strong sense.

\subsection{Positive Result}

We first show a positive result when all groups but one contain a single agent. This is a generalization of the corresponding result for two groups (Theorem \ref{thm:manyone}). The algorithm also uses the round-robin algorithm as a subroutine, but more care must be taken to account for the extra groups. Recall that when all groups contain a single agent, the best known approximation ratio is $3/4$, obtained by Ghodsi et al.'s algorithm \cite{GhodsiHaSe17}.

\begin{theorem}
\label{thm:manyones}
Let $n_1\geq 2$ and $n_2=n_3=\dots=n_k=1$. Then it is possible to give every agent at least $\frac{2}{n_1+2k-3}$ of her maximin share.
\end{theorem}

\begin{proof}
Let $\alpha:=\frac{1}{n_1+2k-3}$. If some agent in a singleton group values a good at least $\alpha$ times her value for the whole set of goods, put that good as the only good in her allocation. Since her maximin share is at most $1/k\leq 1/2$ times her value for the whole set of goods, this agent obtains the desired guarantee. We will give the remaining agents their guarantees with respect to the reduced set of goods and agents. By Lemma \ref{lem:addonemfs}, the maximin share of an agent can only increase as we remove an agent and a good, and the approximation ratio $\frac{2}{n_1+2k-3}$ also increases as $k$ decreases. This implies that guarantees for the reduced instance also translate to ones for the original instance. We recompute each agent's value for the whole set of goods as well as the number of groups and repeat this step until no agent in a singleton group values a good at least $\alpha$ times her value for the whole set of goods.

Next, we normalize the utility of each agent for the whole set of goods to 1 as in Theorem \ref{thm:tworoundrobin}. We let each of the agents in the first group, in arbitrary order, take a good worth at least $\alpha$ to her if there is any. The approximation guarantee is satisfied for any agent who takes a good. Suppose that after this step, there are $n_0\leq n_1$ agents in the first group and $k_0\leq k-1$ agents in the remaining groups who have not taken a good. We run the round-robin algorithm on these $n_0+k_0$ agents, starting with the $k_0$ agents who do not belong to the first group.

Consider one of the $n_0$ agents in the first group. Since no good allocated by the round-robin algorithm is worth at least $\alpha$ to her, she has envy at most $\alpha$ toward each of the $k_0$ agents. Assume for contradiction that the bundle allocated to the first group is worth less than $\alpha$ to her. Then she values the bundle of each of the $k_0$ agents at most $2\alpha$. Hence her utility for the whole set of goods is less than $\alpha+2k_0\alpha=\frac{2k_0+1}{n_1+2k-3}\leq\frac{2k-1}{2k-1}=1$, a contradiction.

Consider now one of the $k_0$ agents in the remaining group. She has utility at least $1-(n_1-n_0)\alpha$ for the set of goods allocated by the round-robin algorithm. With respect to the bundles allocated by the round-robin algorithm, she has no envy toward herself or any of the $n_0$ agents in the first group, and she has envy at most $\alpha$ toward the remaining $k_0-1$ agents. Summing up the corresponding inequalities, averaging, and using the fact that her utility for all bundles combined is at least $1-(n_1-n_0)\alpha$, we find that her utility for her own bundle is at least $\frac{1-(n_1-n_0)\alpha-(k_0-1)\alpha}{n_0+k_0}$. It suffices to show that this is at least $\alpha$. The inequality is equivalent to $\alpha(2k_0+n_1-1)\leq 1$, which holds since $k_0\leq k-1$. \qed
\end{proof}

\subsection{Negative Result}

We next show that when all groups contain at least two agents and one group contains at least five agents, no approximation is possible.

\begin{theorem}
\label{thm:manytwos}
Let $n_1\geq 4$ if $k$ is even and $n_1\geq 5$ if $k$ is odd, and $n_2=n_3=\dots=n_k=2$. Then there exists an instance in which some agent with nonzero maximin share necessarily receives zero utility.
\end{theorem}

\begin{proof}
Let $n_1=4$ if $k$ is even and $5$ if $k$ is odd, and suppose that there are $2k$ goods. In each of the groups $2,3,\dots,k$, one agent has utility 1 for the first $k$ goods and 0 for the last $k$, while the other agent has utility 0 for the first $k$ goods and 1 for the last $k$. Hence all of these agents have a maximin share of 1. To ensure that they all get nonzero utility, each group must receive one of the first $k$ and one of the last $k$ goods. This only leaves one good from the first $k$ and one from the last $k$ to the first group.

First, consider the case $k$ even. Let the utilities of the agents in the first group be given by
\begin{itemize}
\item $\textbf{u}_{11}=(\overbrace{1,1,\dots,1}^{k/2},\overbrace{0,0,\dots,0}^{k/2},\overbrace{1,1,\dots,1}^{k/2},\overbrace{0,0,\dots,0}^{k/2})$;
\item $\textbf{u}_{12}=(\overbrace{1,1,\dots,1}^{k/2},\overbrace{0,0,\dots,0}^{k/2},\overbrace{0,0,\dots,0}^{k/2},\overbrace{1,1,\dots,1}^{k/2})$;
\item $\textbf{u}_{13}=(\overbrace{0,0,\dots,0}^{k/2},\overbrace{1,1,\dots,1}^{k/2},\overbrace{1,1,\dots,1}^{k/2},\overbrace{0,0,\dots,0}^{k/2})$;
\item $\textbf{u}_{14}=(\overbrace{0,0,\dots,0}^{k/2},\overbrace{1,1,\dots,1}^{k/2},\overbrace{0,0,\dots,0}^{k/2},\overbrace{1,1,\dots,1}^{k/2})$.
\end{itemize}
All four agents have a maximin share of 1, but for any combination of a good from the first $k$ goods and one from the last $k$, some agent obtains a utility of 0.

Next, consider the case $k$ odd. Let the utilities of the agents in the first group be given by
\begin{itemize}
\item $\textbf{u}_{11}=(\overbrace{1,1,\dots,1}^{(k-1)/2},\overbrace{0,0,\dots,0}^{(k+1)/2},\overbrace{1,1,\dots,1}^{(k+1)/2},\overbrace{0,0,\dots,0}^{(k-1)/2})$;
\item $\textbf{u}_{12}=(\overbrace{1,1,\dots,1}^{(k+1)/2},\overbrace{0,0,\dots,0}^{(k-1)/2},\overbrace{0,0,\dots,0}^{(k+1)/2},\overbrace{1,1,\dots,1}^{(k-1)/2})$;
\item $\textbf{u}_{13}=(\overbrace{0,0,\dots,0}^{(k+1)/2},\overbrace{1,1,\dots,1}^{(k-1)/2},\overbrace{0,0,\dots,0}^{(k-1)/2},\overbrace{1,1,\dots,1}^{(k+1)/2})$;
\item $\textbf{u}_{14}=(\overbrace{0,0,\dots,0}^{(k-1)/2},\overbrace{1,1,\dots,1}^{(k+1)/2},\overbrace{1,1,\dots,1}^{(k-1)/2},\overbrace{0,0,\dots,0}^{(k+1)/2})$;
\item $\textbf{u}_{15}=(\overbrace{1,1,\dots,1}^{(k-1)/2},0,\overbrace{1,1,\dots,1}^{k-1},0,\overbrace{1,1,\dots,1}^{(k-1)/2})$.
\end{itemize}
All five agents have a maximin share of 1, but as in the previous case, any combination of a good from the first $k$ goods and one from the last $k$ yields no utility to some agent. \qed
\end{proof}

\section{Conclusion and Future Work}

In this paper, we study the problem of approximating the maximin share when we allocate goods to groups of agents. When there are two groups, we characterize the cardinality of the groups for which we can obtain a positive approximation of the maximin share. We also show positive and negative results for approximation when there are several groups.

We conclude the paper by listing some future directions. For two groups, closing the gap between the lower and upper bounds of the approximation ratios (Table \ref{table:twogroups}) is a significant problem from a theoretical point of view but perhaps even more so from a practical one. In particular, it would be especially interesting to determine the asymptotic behavior of the best approximation ratio when one group contains a single agent and the number of agent in the other group grows. For the case of several groups, one can ask whether it is in general possible to obtain a positive approximation when some groups contain a single agent while others contain two agents; the techniques that we present in this paper do not seem to extend easily to this case. Another question is to determine whether the dependence on the number of groups in the approximation ratio (Theorem \ref{thm:manyones}) is necessary. One could also address the issue of truthfulness or add constraints on the allocation, for example by requiring that the allocation form a contiguous block on a line, as has been done for the traditional fair division setting \cite{AmanatidisBiMa16,BouveretCeEl17,Suksompong17}. 

In light of the fact that the positive results in this paper only hold for groups with a small number of agents, a natural question is whether we can relax the fairness notion in order to allow for more positive results. For example, one could consider only requiring that a certain fraction of agents in each group, instead of all of them, think that the allocation is fair. Indeed, for two groups with any number of agents, there exists an allocation that yields at least half of the maximin share to at least half of the agents in each group \cite{SegalhaleviSu17}. Alternatively, if we use envy-freeness as the fairness notion, then a possible relaxation is to require envy-freeness only up to some number of goods, where the number of goods could depend on the number of agents in each group.

From a broader point of view, an intriguing future direction is to explore other fairness notions such as envy-freeness and proportionality in our setting where goods are allocated to groups of agents. Even though an allocation satisfying these notions does not necessarily exist, and indeed even an approximation cannot be guaranteed, we can still strive for an algorithm that produces such an allocation whenever one exists. This has been obtained for envy-freeness for groups with single agents, for example by the undercut procedure \cite{BramsKiKl12}. Such an algorithm for our generalized setting would be both interesting and important in our opinion.

\subsubsection*{Acknowledgments.} The author thanks the anonymous reviewers for their helpful feedback and acknowledges support from a Stanford Graduate Fellowship.

\newpage



\end{document}